\documentclass[12pt]{article}
\usepackage[utf8]{inputenc}
\usepackage[T1]{fontenc}

\usepackage{amsmath, amsthm}
\usepackage{amssymb}

\newtheorem{theorem}{Theorem}[section]
\newtheorem{proposition}[theorem]{Proposition}
\newtheorem{corollary}[theorem]{Corollary}
\newtheorem{lemma}[theorem]{Lemma}

\title{A Unique Extension of Rich Words}

\author{Josef Rukavicka\thanks{Department of Mathematics,
Faculty of Nuclear Sciences and Physical Engineering, CZECH TECHNICAL UNIVERSITY
IN PRAGUE
(josef.rukavicka@seznam.cz).}}

\newtheorem{definition}[theorem]{Definition}
\theoremstyle{remark}
\newtheorem{example}[theorem]{Example}
\newtheorem{remark}[theorem]{Remark}

\DeclareMathOperator{\ltrim}{ltrim}
\DeclareMathOperator{\rtrim}{rtrim}
\DeclareMathOperator{\rw}{R}
\DeclareMathOperator{\lps}{lps}
\DeclareMathOperator{\lpps}{lpps}
\DeclareMathOperator{\lpp}{lpp}
\DeclareMathOperator{\lppp}{lppp}
\DeclareMathOperator{\Prefix}{Prf}
\DeclareMathOperator{\Suffix}{Suf}
\DeclareMathOperator{\SuffixUnion}{SufU}
\DeclareMathOperator{\swtop}{S}
\DeclareMathOperator{\StdExtL}{EL}
\DeclareMathOperator{\StdExtR}{ER}
\DeclareMathOperator{\StdExtLa}{EL_a}
\DeclareMathOperator{\StdExtRa}{ER_a}
\DeclareMathOperator{\StdExtWP}{ewp} 
\DeclareMathOperator{\StdExtOLP}{elpp} 
\DeclareMathOperator{\Alphabet}{A}
\DeclareMathOperator{\AlphabetCard}{\vert A\vert}
\DeclareMathOperator{\Factor}{F}
\DeclareMathOperator{\occur}{occur}
\DeclareMathOperator{\ggt}{g}
\DeclareMathOperator{\reduced}{rdc}
\DeclareMathOperator{\ggh}{h}
\DeclareMathOperator{\switch}{sw}
\DeclareMathOperator{\switchSuf}{swSuf}
\DeclareMathOperator{\swc}{spc} 

\DeclareMathOperator{\flexpref}{T}
\DeclareMathOperator{\rext}{K}
\DeclareMathOperator{\maxPow}{maxPow}

\date{\small{October 01, 2019}\\
   \small Mathematics Subject Classification: 68R15}

\begin{document}
\maketitle

\begin{abstract}
A word $w$ is called rich if it contains $\vert w\vert+1$ palindromic factors, including the empty word. We say that a rich word $w$ can be extended in at least two ways if there are two distinct letters $x,y$ such that $wx,wy$ are rich.

Let $\rw$ denote the set of all rich words. Given $w\in \rw$, let $\rext(w)$ denote the set of all words such that if $u\in \rext(w)$ then $wu\in \rw$ and $wu$ can be extended in at least two ways. Let $\omega(w)=\min\{\vert u\vert \mid u\in \rext(w)\}$ and let $\phi(n)=\max\{\omega(w)\mid w\in \rw\mbox{ and }\vert w\vert=n\}$, where $n>0$.
Vesti (2014) showed that $\phi(n)\leq 2n$. In other words, it says that for each $w\in \rw$ there is a word $u$ with $\vert u\vert\leq 2\vert w\vert$ such that $wu\in \rw$ and $wu$ can be extended in at least two ways. 

We prove that $\phi(n)\leq n$. 
In addition we prove that for each real constant $c>0$ and each integer $m>0$ there is $n>m$ such that $\phi(n)\geq (\frac{2}{9}-c)n$. The results hold for each finite alphabet having at least two letters.
\end{abstract}

\section{Introduction}
A word is called a \emph{palindrome} if it is equal to its reversal. Two examples of palindromes are ``noon'' and ``level''. It is known that a word $w$ can contain at most $\vert w\vert +1$ distinct palindromic factors, including the empty word \cite{DrJuPi}. If the bound $\vert w\vert+1$ is attained, the word $w$ is called \emph{rich}. Quite many articles investigated the properties of rich words in recent years, for example \cite{BuLuGlZa2, DrJuPi, GlJuWiZa, RukavickaRichWords2019, Vesti2014}. Some of the properties of rich words are stated in the next section; see Propositions \ref{yy65se85bj5d}, \ref{kkmnd5s8s658}, and \ref{oij5498fr654td222gh}.

In \cite{GlJuWiZa} it was proved that if $w$ is rich then there is a letter $x$ such that $wx$ is also rich.  
In \cite{Vesti2014} it was proved that if $w$ is rich then there is a word $u$ and two distinct letters $x,y$ such that $\vert u\vert\leq 2\vert w\vert$ and $wux, wuy$ are rich. Concerning this result, the author of \cite{Vesti2014} formulated an open question: \begin{itemize}\item Let $w$ be a rich word. How long is the shortest $u$ such that $wu$ can always be extended in at least two ways?\end{itemize} 

In the current article we improve the result from \cite{Vesti2014} and as such, to some extent, we answer to the open question. Let $\rw$ denote the set of all rich words. We say that a rich word $w$ \emph{can be extended in at least two ways} if there are two distinct letters $x,y$ such that $wx,wy$ are rich. Given $w\in \rw$, let $\rext(w)$ denote the set of all words such that if $u\in \rext(w)$ then $wu\in \rw$ and $wu$ can be extended in at least two ways; $\rext(w)$ contains the empty word if $w$ can be extended in at least two ways. Let $\omega(w)=\min\{\vert u\vert \mid u\in \rext(w)\}$ and let $\phi(n)=\max\{\omega(w)\mid w\in \rw\mbox{ and }\vert w\vert=n\}$, where $n>0$. 
The result from \cite{Vesti2014} can be presented as $\phi(n)\leq 2n$.

We show that $\phi(n)\leq n$. It is natural to ask how good this bound is. The rich word $wu$ is called a \emph{unique rich extension} of $w$ if there is no proper prefix $\bar u$ of $u$ such that $w\bar u$ can be extended in at least two ways. In Remark $2.4$ in \cite{Vesti2014} there is an example which shows that there are $w_n,u_n\in \rw$ such that $w_nu_n$ is a unique rich extension of $w_n$ and $\vert u_n\vert=n$, where $n>1$. However in the given example the length of $w_n$ grows significantly more rapidly than the length of $u_n$ as $n$ tends towards infinity. This could suggest that $\lim_{n\rightarrow\infty}\frac{\phi(n)}{n}=0$; we show that this suggestion is false. We prove that for each real constant $c>0$ and each integer $m>0$ there is $n>m$ such that $\phi(n)\geq (\frac{2}{9}-c)n$. 

We explain the idea of the proof. Let $w^R$ denote the reversal of the word $w$.
We construct rich words $\ggh_n=u_nv^Rtv_n$, where $n\geq 3$ such that 
\begin{enumerate} 
\item \label{fjde5649w3d23456}
The word $t$ is the longest palindromic suffix of $u_nv_n^Rt$.
\item
\label{alkd8e3df68e9f59}
For every factor $xpy$ of $tv_n$ we have that $xpx$ is a factor of $u_n$, where $x,y$ are distinct letters and $p$ is a palindrome.
\item \label{jdhkwi55d46d8w}
$2\vert \ggh_n\vert<\vert\ggh_{n+1}\vert$.
\end{enumerate}
Let $\bar vx$ be a prefix of $v_n$, where $x$ is a letter. Let $y$ be a letter distinct from $x$ and let $ypy$ be the longest palindromic suffix of $u_nv_n^Rt\bar vy$. Property \ref{fjde5649w3d23456} implies that $ypy$ is a suffix of $y\bar v^Rt\bar vy$, since $\bar v^Rt\bar v$ is the longest palindromic suffix of $u_nv_n^Rt\bar v$. 
Property \ref{alkd8e3df68e9f59} implies that $ypy$ is not unioccurrent in $u_nv_n^Rt\bar vy$. In consequence $u_nv_n^Rt\bar vy$ is not rich; see Proposition \ref{oij5498fr654td222gh}. Hence there is no proper prefix $v$ of $v_n$ such that $u_nv_n^Rtv$ can be extended in at least two ways. It follows that $\vert v_n\vert \leq \omega(u_nv^Rt)$. Property \ref{jdhkwi55d46d8w} implies that for each $m>0$ there is $n$ such that $\vert\ggh_n\vert>m$.

We will see that to find $u_n$ for given $v_n$ is quite straightforward. The crucial part of our construction is the word $v_n$. To be specific, the word $v_n$ that we will present contains only a ``small'' number of factors $xpy$ defined in Property \ref{alkd8e3df68e9f59}. As a result the length of $u_n$ grows almost linearly with the length of $v_n$ as $n$ tends towards infinity.

\section{Preliminaries}
Consider an alphabet $A$ with $q$ letters, where $q>1$. Let $A^+$ denote the set of all nonempty words over $A$. Let $\epsilon$ denote the empty word, and let $A^*=A^+\cup \{\epsilon\}$. We have that $\rw\subseteq \Alphabet^*$.

Let $\Factor(w)$ be the set of all factors of the word $w\in A^*$; we define that $\epsilon, w\in \Factor(w)$. Let $\Prefix(w)$ and $\Suffix(w)$ be the set of all prefixes and all suffixes of $w\in A^*$ respectively; we define that $\{\epsilon, w\}\subseteq \Prefix(w)\cap\Suffix(w)$. 

Let $\SuffixUnion(v,u)=\bigcup_{t\in \Prefix(u)\setminus\{\epsilon\}}\Suffix(vt)$, where $v,u\in \Alphabet^*$. The set $\SuffixUnion(v,u)$ is the union of sets of suffixes of $vt$, where $t$ is a nonempty prefix of $u$.

We define yet the reversal that we have already used in the introduction: Let $w^R$ denote the reversal of $w\in A^*$; formally if $w=w_1w_2\dots w_k$ then $w^R=w_k\dots w_2w_1$, where $w_i\in \Alphabet$ and $i\in \{1,2,\dots,k\}$. 

Let $\lps(w)$ and $\lpp(w)$ denote the longest palindromic suffix and the longest palindromic prefix of $w\in \Alphabet^*$ respectively. We define that $\lps(\epsilon)=\lpp(\epsilon)=\epsilon$. Let $\lpps(w)$ and $\lppp(w)$ denote the longest proper palindromic suffix and the longest proper palindromic prefix of $w\in \Alphabet^*$ respectively, where $\vert w\vert\geq 1$. If $\vert w\vert=1$ then we define $\lppp(w)=\lpps(w)=\epsilon$.

Let $\rtrim(w)=v$, where $v,w\in \Alphabet^*$, $y\in \Alphabet$, $w=vy$, and $\vert w\vert\geq 1$.
Let $\ltrim(w)=v$ , where $v,w\in \Alphabet^*$, $x\in \Alphabet$, $w=xv$, and $\vert w\vert\geq 1$. The functions $\rtrim(w)$ and $\ltrim(w)$ remove the last and the first letter of $w$ respectively.

Let $\occur(u,v)$ be the number of occurrences of $v$ in $u$, where $u,v\in \Alphabet^+$; formally $\occur(u,v)=\vert \{w\mid w\in \Suffix(u)\mbox{ and }v\in \Prefix(w)\}\vert$. We call a factor $v$ \emph{unioccurrent} in $u$ if $\occur(u,v)=1$.

We list some known properties of rich words that we use in our article. All of them can be found, for instance, in  \cite{GlJuWiZa}.
Recall the notion of a \emph{complete return} \cite{GlJuWiZa}: Given a word $w$ and factors $r,u\in \Factor(w)$, we call the factor $r$ a complete return to $u$ in $w$ if $r$ contains exactly two occurrences of $u$, one as a prefix and one as a suffix. 
\begin{proposition}
\label{yy65se85bj5d}
If $w,u\in\rw\cap \Alphabet^+$, $u\in \Factor(w)$, and $u$ is a palindrome then all complete returns to $u$ in $w$ are palindromes.
\end{proposition}
\begin{proposition}
\label{kkmnd5s8s658}
If $w \in \rw$ and $p\in \Factor(w)$ then $p, p^R\in \rw$.
\end{proposition}
\begin{proposition}
\label{oij5498fr654td222gh}
A word $w$ is rich if and only if every prefix $p\in \Prefix(w)$ has a unioccurrent palindromic suffix.
\end{proposition}
From Proposition \ref{kkmnd5s8s658} and Proposition \ref{oij5498fr654td222gh} we have an obvious corollary.
\begin{corollary}
\label{odd58e62n2so}
A word $w$ is rich if and only if every suffix $p\in \Suffix(w)$ has a unioccurrent palindromic prefix.
\end{corollary}

\section{Standard Extension}
We define a \emph{left standard extension} and a \emph{right standard extension} of a rich word. The construction of a standard extension has already been used in \cite{Vesti2014}. The name ``standard extension'' has been introduced later in \cite{RukavickaRichWords2019}.  Here we use a different notation and we distinguish a left and a right standard extension. 
\begin{definition}
Let $j\geq 0$ be a nonnegative integer, $w\in \rw$, and $\vert w\vert\geq 1$. We define $\StdExtR^j(w)$, $\StdExtL^j(w)$ as follows:
\begin{itemize}
\item
$\StdExtR^0(w)=\StdExtL^0(w)=w$.
\item
$\StdExtL(w)=\StdExtL^1(w)=xw$, where $x\in \Alphabet$ is such that $\lppp(w)x\in \Prefix(w)$.
\item
$\StdExtR(w)=\StdExtR^1(w)=wx$, where $x\in \Alphabet$ is such that $x\lpps(w)\in \Suffix(w)$.
\item
$\StdExtL^j(w)=\StdExtL(\StdExtL^{j-1}(w))$, where $j>1$.
\item
$\StdExtR^j(w)=\StdExtR(\StdExtR^{j-1}(w)$, where $j>1$.
\end{itemize}
Let $\StdExtLa(w)=\{\StdExtL^j(w)\mid j\geq 0\}$. We call $p\in\StdExtLa(w)$ a left standard extension of $w$. 
Let $\StdExtRa(w)=\{\StdExtR^j(w)\mid j\geq 0\}$. We call $p\in \StdExtRa(w)$ a right standard extension of $w$. 
\end{definition}
\begin{remark}
It is easy to see that $\StdExtR^j(w)=(\StdExtL^j(w^R))^R$ and $\StdExtL^j(w)=(\StdExtR^j(w^R))^R$, where $j\geq 0$.

If $x\in \Alphabet$ then $\StdExtR(x)=\StdExtL(x)=xx$, since $\lppp(x)=\lpps(x)=\epsilon$.
\end{remark}
\begin{example}
Let $\Alphabet=\{0,1,2,3\}$ and $w=010200330$. Then we have:
\begin{itemize}
\item
$\lppp(w)=010$ and $\lpps(w)=0330$.
\item
$\StdExtR(w)=0102003300$, $\StdExtR^2(w)=01020033002$, \\ $\StdExtR^3(w)=010200330020$, $\StdExtR^4(w)=0102003300201$, \\ $\StdExtR^5(w)=01020033002010$, $\StdExtR^6(w)=010200330020102$, \\ $\StdExtR^7(w)=0102003300201020$.
\item
$\StdExtL(w)=2010200330$, $\StdExtL(w)^2=02010200330$, \\ $\StdExtL(w)^3=002010200330$, $\StdExtL(w)^4=3002010200330$, \\ $\StdExtL(w)^5=33002010200330$, $\StdExtL(w)^6=033002010200330$, \\ $\StdExtL(w)^6=0033002010200330$, $\StdExtL(w)^7=20033002010200330$.
\end{itemize}
\end{example}
A left and a right standard extension of a rich word $w$ is rich. In consequence, every rich word $w$ can be extended to rich words $wx,yw$ for some letters $x,y$; this has already been proved in \cite{GlJuWiZa, RukavickaRichWords2019, Vesti2014}.
\begin{lemma}
\label{fjis5e64s8e6jj5}
If $w\in \rw$ and $\vert w\vert\geq 1$ then $\StdExtRa(w)\cup\StdExtLa(w)\subseteq \rw$.
\end{lemma}
\begin{proof}
Since $\StdExtL^j(w)=(\StdExtR^j(w^R))^R$ and since for every $t\in \StdExtRa(w)\setminus\{w\}$ there is a rich word $\bar t\in \StdExtRa(w)$ such that $t=\StdExtR(\bar t)$, it is enough to prove that $\StdExtR(w)\in \rw$.
 
Let $xpx=\lps(\StdExtR(w))$, where $x\in \Alphabet$. Because $w\in \rw$, Proposition \ref{oij5498fr654td222gh} implies that we need to prove that $xpx$ is unioccurrent in $\StdExtR(w)$. Realize that $p=\lpps(w)$; it means that $p$ is either unioccurrent in $w$ or $w$ is a complete return to $p$. In either case $xpx$ is unioccurrent in $\StdExtR(w)$. This completes the proof.
\end{proof}

\section{A unique rich extension}
We formally define a unique rich extension mentioned in the introduction. In addition we define a \emph{flexed point} of a rich word.
\begin{definition} 
If $u,v\in \rw\cap \Alphabet^+$, $v\in \Prefix(u)$, and \[\Prefix(\rtrim(u))\cap \{vt\mid t\in \omega(v)\}=\emptyset\] then we call $u$ a \emph{unique rich extension} of $v$.

Given $v\in \rw$ with $\vert v\vert>1$, let \[\flexpref(v)=\{ux\mid ux\in \Prefix(v)\mbox{ and }x\in \Alphabet\mbox{ and }ux\not= \StdExtR(u)\}\mbox{.}\] We call $w\in \flexpref(v)$ a flexed point of $v$. 
\end{definition}
\begin{remark} 
Note that if $x\in \Alphabet$ and $ux$ is a flexed point of a rich word $v$ then $u$ can be extended in at least two ways. A similar notion of a ``flexed palindrome'' has been used in \cite{RukavickaRichWords2019}.
\end{remark}
\begin{example}
Let $\Alphabet=\{0,1,2\}$. 
\begin{itemize}
\item
The rich word $00101$ can be extended in at least two ways, because $001010$, $001011$, and $001012$ are rich.
\item
The rich word $20010110$ cannot be extended in at least two ways because $200101100$ and $200101102$ are not rich. Only the right standard extension $200101101$ is rich. Hence $200101101$ is a unique rich extension of $20010110$. 
\item
If $w=201011011101111011111001$ then $w1111$ is unique rich extension of $w$; this example is a modification of the example in Remark $2.4$ in \cite{Vesti2014}.
\item
If  $w=2010110111011110111$ then the set of flexed points of $w$ is: \[\begin{split}\flexpref(w)=\{20,201,20101,201011,2010110111, \\ 20101101110111, 201011011101111\}\mbox{.}\end{split}\]
\end{itemize}
\end{example}
There is a connection between a unique rich extension and a right standard extension. 
\begin{lemma}
\label{ujdd5js666e566d5}
If $u$ is a unique rich extension of $w$ then $u\in \StdExtRa(w)$.
\end{lemma}
\begin{proof}
Suppose there is $\bar ux\in \Prefix(u)$ such that $\bar u\in \StdExtRa(w)$, $x\in \Alphabet$, and $\bar ux\not \in \StdExtRa(w)$. Then obviously $\bar u$ can be extended in at least two ways, since both $\bar ux$ and $\StdExtR(\bar u)$ are rich. Hence $u$ cannot be a unique rich extension of $w$. The lemma follows.
\end{proof}

To simplify the formulation of next lemmas and propositions concerning a unique rich extension we define an auxiliary set $\Gamma$ as follows: $(v,\bar v,u)\in \Gamma$ if $v\bar vu$ is a unique rich extension of $v\bar v$ and $\lpps(v\bar v)=\bar v$, where 
$v,\bar v, u\in \rw\cap\Alphabet^+$.

We show that if $wu$ is unique rich extension of $w$, then $\lpps(w)$ is unioccurrent in $\lpps(w)u$. 
\begin{proposition}
\label{uuhj54t58p}
If $(v,\bar v,u)\in \Gamma$ 
then $\occur(\bar vu,\bar v)=1$.
\end{proposition}
\begin{proof}
The proposition follows from the proof of Theorem $2.1$ in \cite{Vesti2014}.
The author shows that a rich word $w$ can be extended into a rich word $w\bar w$ in such a way that $a^n$ is a suffix of $w\bar w$, where $a^n$ is the largest power of some letter $a\in \Alphabet$. It is proved that $w\bar w$ can be extended in at least two ways. In both cases distinguished in the proof of Theorem $2.1$ in \cite{Vesti2014} it is easy to see that $\occur(\lpps(w)\bar w, \lpps(w))=1$. The proposition follows. 
\end{proof}

We present two simple properties of a unique rich extension.
\begin{lemma}
\label{ujdd54fvb1a2w}
Let $(v,\bar v,u)\in \Gamma$.
\begin{enumerate}
\item
If  $\vert u\vert \leq\vert v\vert$ then $u^R \in \Suffix(v)$.
\item
If $\vert u\vert \geq\vert v\vert$ then $v^R \in \Prefix(u)$.
\end{enumerate}
\end{lemma}
\begin{proof}
Obviously $v\bar vv^R\in \StdExtRa(v\bar v)$.  Lemma \ref{ujdd5js666e566d5} implies that $v\bar vu\in \StdExtRa(v\bar v)$. The lemma follows.
\end{proof}
The next proposition discusses  words of the form $v\bar vux$, where $v\bar vux$ is unique rich extension of $v\bar v$, $x$ is a letter, $\bar v$ is the longest proper palindromic suffix of $v\bar v$, and $\bar vux$ is a flexed point of  $\bar vux$. The proposition asserts that there are words $t_1, t_2$ such that $v=t_1t_2$, $xu^R$ is a proper suffix of $t_2$, and $\bar vt_2^R$ is a flexed point of $\bar vt_2^R$. In particular it implies that $\vert v\vert >\vert ux\vert$.
\begin{proposition}
\label{ppph522g1k}
If $(v,\bar v,ux)\in \Gamma$ and $\bar vux\in \flexpref(\bar vux)$ then there exist $t_1,t_2\in \rw$ such that
\begin{itemize}
\item 
$v=t_1t_2$,
\item
$xu^R\in \Suffix(\ltrim(t_2))$, and
\item 
$\bar vt_2^R\in \flexpref(\bar vt_2^R)$.
\end{itemize}
\end{proposition}	
\begin{proof}
Let $w=\lpps(\bar vu)$ and let $y\in \Alphabet$ be such that  $yw\in \Suffix(\bar vu)$. Since  $\bar vux\in \flexpref(\bar vux)$ we have that $x\not=y$.

Obviously $ywy\in \Factor(v\bar vu)$ because $v\bar vux$ is a unique rich extension of $v\bar v$ and thus $v\bar vuy\not \in \rw$. Hence the palindromic suffix $ywy$ of $v\bar vuy$ is not unioccurrent in $v\bar vuy$, see Proposition \ref{oij5498fr654td222gh}. 

We have that $w$ is unioccurrent in $\bar vu$ and $\bar v\not \in \Factor(w)$, since $w=\lpps(\bar vu)$ and $\bar v$ is unioccurrent in $\bar vu$, see Proposition \ref{uuhj54t58p}. It follows that there are $t_1,t_2\in \Factor(v)$ such that $v=t_1t_2$, $ywy\in \Prefix(t_2\bar vux)$ and $ywy$ is unioccurrent in $t_2\bar vux$. Thus $\lpp(yt_2\bar vux)=ywy$. 

From the fact that $\bar v\not \in \Factor(w)$ follows that $ywy\in \Prefix(t_2\bar v)$. Lemma \ref{ujdd54fvb1a2w} implies that $\vert t_2\vert\geq \vert ux\vert$ and $xu^R\in \Suffix(\ltrim(t_2))$. Just consider that $\vert t_2\vert\leq \vert ux\vert$ would imply that $ywy\in \Factor(\bar vux)$.

Since $xu^R\bar v\in \Suffix(t_2\bar v)$, $w\in \Suffix(\bar vu)$, $ywy\in \Prefix(t_2\bar v)$, and $x\not=y$ it follows that $\occur(t_2\bar v,w)>1$; hence Proposition \ref{yy65se85bj5d} implies that $\lppp(\ltrim(t_2\bar v))\not=w$. It follows that $t_2\bar v\not=\StdExtL(\ltrim(t_2\bar v))$. Consequently $\bar vt_2^R\not=\StdExtR(\rtrim(\bar vt_2^R))$ and thus $\bar vt_2^R\in \flexpref(\bar vt_2^R)$. This completes the proof.
\end{proof}
We step to the main result of this section. The theorem says that if $v\bar vu$ is a unique rich extension of $v\bar v$ and $\bar v$ is the longest proper palindromic suffix of $v\bar v$ then $u$ is not longer than $v\bar v$.
\begin{theorem}
\label{tmlsk566e58t5h4}
If $(v,\bar v,u)\in \Gamma$ then $\vert u\vert\leq \vert v\bar v\vert$.
\end{theorem}
\begin{proof}
Let $(v,\bar v,u)\in \Gamma$.
If $\vert u\vert +\vert \lpps(\bar v)\vert\leq\vert \bar v\vert$ then clearly $\vert u\vert\leq \vert v\bar v\vert$.
For the rest of the proof suppose that $\vert u\vert +\vert \lpps(\bar v)\vert>\vert \bar v\vert$.
We show that the set of flexed points $\flexpref(\bar vu)$ is nonempty. Let $\bar v=h\lpps(\bar v)$. Proposition \ref{uuhj54t58p} implies that $h^R\not \in \Prefix(u)$, because $\occur(h\lpps(\bar v)h^R, \bar v)=2$. Since $\vert u\vert +\vert \lpps(\bar v)\vert>\vert \bar v\vert$ it follows that there are $\bar u\in\rw$ and $x\in \Alphabet$ such that $\bar ux\in \Prefix(u)$ and $\bar v\bar ux\not=\StdExtR(\bar v\bar u)$; just realize that $h\lpps(\bar v)h^R\in \StdExtRa(h\lpps(\bar v))$. We showed that $\flexpref(\bar vu)\setminus\Prefix(\bar v)\not =\emptyset$.

Without lost of generality, suppose that $\bar v\bar ux$ is the longest flexed point from the set $\flexpref(\bar vu)\setminus\Prefix(\bar v)$ and suppose that $\vert u\vert >\vert v\vert$. Proposition \ref{ppph522g1k} asserts that there are $t_1,t_2\in \rw$ such that $v=t_1t_2$, $\bar vt_2^R\in \flexpref(\bar vt_2^R)$, and $x\bar u^R\in \Suffix(\ltrim(t_2))$. If $\vert u\vert>\vert v\vert$, then $\bar vt_2^R\in \Prefix(\bar vu)$, see Lemma \ref{ujdd54fvb1a2w}. This is a contradiction, since we supposed that $\bar v\bar ux$ is the longest flexed point of $\bar vu$. We conclude that $\vert u\vert \leq \vert v\vert$. This completes the proof.
\end{proof}
The simple corollary is that if $wu$ is a unique rich extension of $w$ then $u$ is not longer than $w$.
\begin{corollary}
If $n\geq 1$ then $\phi(n)\leq n$.
\end{corollary}
\begin{proof}
The corollary is obvious for $n\in \{1,2\}$. If $wu$ is a unique rich extension of $w$, $\vert w\vert\geq 2$, and $\vert u\vert \geq 1$ then there is clearly $(v,\bar v, u)\in \Gamma$ such that $w=v\bar v$. Then the corollary follows from Theorem \ref{tmlsk566e58t5h4}.
\end{proof}
\section{Construction of a Uniquely Extensible Rich Word I}

\begin{definition}
We call a word $xpy$ a \emph{switch} if $x,y\in \Alphabet$, $x\not=y$, and $p\in \Alphabet^*$ is a palindrome.  
Let $\switch(v)=\{w\mid w\in \Factor(v)\mbox{ and }w\mbox{ is a switch}\}$. Let $\switchSuf(v,u)=\switch(vu)\cap\SuffixUnion(v,u)$, where $v,u\in \Alphabet^*$.

Given $S\subseteq \Alphabet^*$, let \[\reduced(S)=\{w\mid w\in S\mbox{ and }w\not \in \bigcup_{u\in S\setminus\{w\}}\Factor(u)\}\mbox{.}\] We call $\reduced(S)$ a \emph{reduction} of $S$.

Suppose $xpy$ is a switch, let $\swc(xpy)=xpx$, where $x,y\in \Alphabet$. We call $\swc(xpy)$ a \emph{switch palindromic closure} of the switch $xpy$. If $B\subset \Alphabet^+$ is a set of switches then we define $\swc(B)=\reduced(\bigcup_{w\in B}\{\swc(w)\})$. 

\end{definition}
\begin{remark}
Note that if $xpy$ is a switch, then $p$ can be the empty word.

The set $\switchSuf(v,u)$ is a set of switches that are suffixes of $v\bar u$ for all nonempty prefixes $\bar u$ of $u$.

The reduction $\reduced(S)$ of the set $S$ is a subset of $S$ and contains only elements that are not proper factors of other elements of $S$.

The switch palindromic closure of a set $B$ is a reduction of the union of all switch palindromic closures of switches from the set $B$.
\end{remark}
\begin{example}
Let $\Alphabet=\{0,1,2\}$, $v=0100110$, and $u=12$. Then we have:
\begin{itemize}
\item
$\switch(vu)=\{01,10,100,110,011,001,010011,001101, 12, 012,11012\}$.
\item
$\switchSuf(v,u)=\left(\switch(v1)\cap\Suffix(v1)\right) \cup \left(\switch(v12)\cap\Suffix(v12)\right)=$\\$\{001101\}\cup \{12, 012,11012\}$.
\item
$\swc(001101)=001100$, $\swc(12)=11$, $\swc(012)=010$, \\ $\swc(11012)=11011$.
\item
$\swc(\switchSuf(v,u))=\reduced(\{001100, 11, 010, 11011\})=$\\$\{001100,010,110011\}$.
\end{itemize}
\end{example}
The following proposition clarifies the importance of switches for a unique rich extension of rich words. The proposition says that if \begin{itemize}\item $wu^R\bar vu$ is a rich word and \item $\bar v$ is the longest palindromic suffix of $wu^R\bar v$ and \item  $x$ is a factor of $w$ for every letter and \item for every switch $t$ which is a suffix of $wu^R\bar v\bar u$ for some $\bar u\in \Prefix(u)$ we have that $\swc(t)$ is a factor of $w$ \end{itemize} then $wu^R\bar vu$ is unique rich extension of $wu^R\bar v$.
\begin{proposition}
\label{tg512c12vcc25e9}
If $w,u,\bar v\in\Alphabet^+$, $wu^R\bar vu\in \rw$, $\lps(wu^R\bar v)=\bar v$, $\Alphabet\cap\Factor(w)=\Alphabet$, and $\swc(\switchSuf(wu^R\bar v,u))\subseteq\Factor(w)$ then $wu^R\bar vu$ is a unique rich extension of $wu^R\bar v$.
\end{proposition}
\begin{proof}
We show that there is no prefix $\bar ux\in \Prefix(u)\cap\omega(wu^R\bar v)$, where $x\in \Alphabet$.
Suppose that there is $\bar ux\in \Prefix(u)\cap\omega(wu^R\bar v)$. Let $y\in \Alphabet$ be such that $x\not=y$ and $wu^R\bar v\bar uy\in \rw$. 
Let $t=\lps(wu^R\bar v\bar uy)$. We distinguish two cases:
\begin{itemize}
\item
$t\in \Alphabet$. The assumptions of the proposition guarantee that $t\in \Factor(w)$.
\item
$t=y\bar ty$ for some palindrome $\bar t$. Clearly $y\bar tx\in \switchSuf(wu^R\bar v,u)$ and the assumptions of the proposition guarantee that $t=\swc(ytx)=yty\in \Factor(w)$. 
\end{itemize}
It follows that the longest palindromic suffix $t$ is not unioccurrent, hence $wu^R\bar v\bar uy$ is not rich; see Proposition \ref{oij5498fr654td222gh}. This completes the proof.
\end{proof}

Given a factor $u$ of a word $w$, for us it will not be important if $u$ or $u^R$ is unioccurrent in $w$. For this purpose we define a special notion.
\begin{definition}
If $\sum_{v\in \{u,u^R\}}\occur(w,v)=1$ then we say that the word $u$ is \emph{reverse-unioccurrent} in $w$, where $w,u\in \Alphabet^+$.
\end{definition}
\begin{remark}
The notion of reverse-unioccurrence has also been used in \cite{RukavickaRichWords2019}.
\end{remark}

We show that if the switch $ytx$ is a suffix of the word $wx$ and $ytx$ is reverse-unioccurrent in $wx$ then $wx$ is a flexed point of $wx$.
\begin{lemma}
\label{tj55dk5e8e8}
If $w, wx\in \rw$, $x,y\in \Alphabet$, $ytx\in \Suffix(wx)\cap\switch(wx)$, and $ytx$ is reverse-unioccurrent in $wx$ then $wx \in \flexpref(wx)$.
\end{lemma}
\begin{proof}
Suppose that $wx\in \StdExtRa(w)$. If $u=\lpps(w)$ then $\vert t\vert<\vert u\vert$ and $t\in \Prefix(u)\cap\Suffix(u)$. It follows that $xux\in \Suffix(wx)$ and $ty\in \Prefix(u)$, since $yt\in \Suffix(u)$. Consequently $xty\in \Prefix(xu)$, which is a contradiction, because $xty$ is reverse-unioccurrent in $wx$. The lemma follows.  
\end{proof}
There is an obvious corollary of Lemma \ref{tj55dk5e8e8} saying that if $t$ is a switch of $w$, then there is a flexed point $v$ of $w$ such that either $t$ or $t^R$ is a suffix of $v$.
\begin{corollary}
\label{mnr5r152riq4}
If $w\in \rw$, $t\in \switch(w)$ then there is $v\in \flexpref(w)$ such that $\{t,t^R\}\cap \Suffix(v)\not =\emptyset$.
\end{corollary}
\begin{proof}
If $w\in \rw$ and $t\in \switch(w)$, then there is obviously $u\in \Prefix(w)$ such that $\{t, t^R\}\cap \Suffix(u)\not=\emptyset$ and $t$ is reverse-unioccurrent in $u$. Then Lemma \ref{tj55dk5e8e8} implies that $u\not\in \StdExtRa(\rtrim(u))$. This completes the proof.
\end{proof}

In order to construct a word with a prefix containing all switch palindromic closures of its switches we introduce two functions $\StdExtWP$ and $\StdExtOLP$.
\begin{definition}
If $w,t \in \rw\cap \Alphabet^+$ and $t$ is a palindrome then we define
\[\Sigma_{w,t}=\{u\mid u\in \Prefix(w)\mbox{ and }\vert u\vert \geq \vert \lppp(w)\vert\mbox{ and }\rtrim(t)\in \Suffix(u)\}\mbox{.}\] 
If $\Sigma_{w,t}\not=\emptyset$ then let $\bar \pi_{w,t}$ denote the shortest element of $\Sigma_{w,t}$ and let $\pi_{w,t}$ be such that $\bar \pi_{w,t}=\lppp(w)\pi_{w,t}$. 

Let $x=\Prefix(t)\cap\Alphabet$ and let
\[
\StdExtWP(w,t)=\begin{cases}
x(\pi_{w,t})^Rw & \mbox{if } \Sigma_{w,t}\not=\emptyset\mbox{ and }t\not\in \Factor(v^Rw)\\
w &  \mbox{otherwise.}
\end{cases}
\]
In addition we define \[\StdExtWP(w,t_1,t_2,\dots, t_m)=\StdExtWP(\dots (\StdExtWP(\StdExtWP(w,t_1),t_2),\dots ),t_m)\mbox{,}\] where $w$ is a nonempty rich word and $t_1,t_2,\dots, t_m$ are rich nonempty palindromes.

Given $w\in \Alphabet^+$ and $x\in \Alphabet$, let $\maxPow(w,x)=k$ such that $x^k\in \Factor(w)$ and $x^{k+1}\not \in \Factor(w)$. 

Suppose $w\in \rw$, $y\in \Alphabet$, and $k=\maxPow(w,y)$.  Let $\StdExtOLP_y(w)=\StdExtWP(w,y^{k+1})$.
\end{definition}
\begin{remark}
The notation ``ewp'' stands for ``extension with prefix''. It is clear that $(\pi_{w,t})^Rw$ is a left standard extension of $w$ that has as a prefix $\ltrim(t)$.

The notation ``maxPow'' stands for ``maximal power''. If $x\not \in \Factor(w)$ then $\maxPow(w,x)=0$.

The notation ``elpp'' stands for ``extension with letter power prefix''. The function $\StdExtOLP_y(w)$ is the word $yu$ where $u$ is a left standard extension of $w$ such that $y^{\maxPow(w,y)}$ is a prefix. If $\maxPow(w,y)=0$ then $\StdExtOLP_y(w)=yw$. 
\end{remark}
\begin{example}
Let $\Alphabet=\{0,1,2\}$, $w=2020111010111010$, $t_1=11011$, and $t_2=20201$. Then we have:
\begin{itemize}
\item
$\rtrim(t_1)=1101$, $\ltrim(t_1)=1011$, $\lpp(w)=202$.
\item
$\Sigma_{w,t_1}=\{202011101, 202011101011101\}$, $\sigma_{w,t_1}=011101$.
\item
$\StdExtWP(w,t_1)=11011102020111010111010$.
\item
Let $v=\StdExtWP(w,t_1)$. Then $\sigma_{v,t_2}=102020$
\item
$\StdExtWP(w,t_1,t_2)=\StdExtWP(v,t_2)=202020111011102020111010111010$.
\item
$\maxPow(w,1)=3$, $\maxPow(w,2)=1$, and $\maxPow(w,0)=1$.
\item
$\StdExtOLP_1(w)=\StdExtWP(w,1111)=111102020111010111010$.
\item
$\StdExtOLP_2(w)=\StdExtWP(w,22)=22020111010111010$.
\item
$\StdExtOLP_0(w)=\StdExtWP(w,00)=002020111010111010$.
\end{itemize}
\end{example}
We prove that $\StdExtWP(w,t), \StdExtOLP_y(w)\in \rw$ are rich words.
\begin{lemma}
\label{ttf56r25d2f}
If $w,t\in \rw\cap \Alphabet^+$ and $y\in \Alphabet$ then $\StdExtWP(w,t), \StdExtOLP_y(w)\in \rw$.
\end{lemma}
\begin{proof}
Because $\StdExtOLP(w)=\StdExtWP_y(w,y^{k+1})$ it suffices to prove that $\StdExtWP(w,t)\in \rw$. From the definition of $\StdExtWP(w,t)$ it is clear that we need to verify only the case where $\Sigma_{w,t}\not=\emptyset$ and $t\not\in \Factor(v^Rw)$. Obviously $(\pi_{w,t})^Rw\in \rw$, since $(\pi_{w,t})^Rw\in \StdExtLa(w)$, see Lemma \ref{fjis5e64s8e6jj5}. Let $x=\Prefix(t)\cap\Alphabet$. Then $\lpp(xv^Rw)=t$ and since $t\not \in \Factor(v^Rw)$ we have $\occur(xv^Rw,t)=1$. Hence Corollary \ref{odd58e62n2so} implies that $xv^Rw\in \rw$.
\end{proof}

\section{Construction of a Uniquely Extensible Rich Word II}

In this section we consider that $\{0,1\}\subseteq\Alphabet$. Let $\ggt_n=\ggt_{n-1}01^n0\ggt_{n-1}$, where $\ggt_1=1$ and $n>1$. For $n,k\geq 2$ we show that the words $0^k\ggt_n$ are rich and that $0^k\ggt_{n-1}01, 0^k\ggt_{n-1}01^{n}$ are the only flexed points of $0^k\ggt_n$ that are not flexed points of $0^k\ggt_{n-1}$. Let $\bar \flexpref_n=\flexpref(0^k\ggt_n)\setminus\flexpref(0^k\ggt_{n-1})$.
\begin{proposition}
\label{hjf568sty}
If $n,k\geq 2$ then $0^kg_n\in \rw$ and
$$\bar \flexpref_n=\{0^k\ggt_{n-1}01, 0^k\ggt_{n-1}01^{n}\}\mbox{.}$$
\end{proposition}
\begin{proof}
Obviously $0^k\ggt_1\in \rw$. Suppose that $0^k\ggt_{n-1}\in \rw$, where $n\geq 2$. We show that $0^k\ggt_{n}\in \rw$. We have that $0^k\ggt_n=0^k\ggt_{n-1}01^n0\ggt_{n-1}$. Note that $\lps(0^k\ggt_{n-1})=\ggt_{n-1}$. It follows that $0^k\ggt_{n-1}0=\StdExtR(0^k\ggt_{n-1})$ and hence $0^k\ggt_{n-1}0\in \rw$. 
It is easy to see that $$\lps(0^k\ggt_{n-1}01)=\lps(0^k\ggt_{n-2}01^{n-1}0\ggt_{n-2}01)=10\ggt_{n-2}01$$ and that $\occur(0^k\ggt_{n-1}01, 10\ggt_{n-2}01)=1$. Hence we 
have $0^k\ggt_{n-1}01\in \rw$; see Proposition \ref{oij5498fr654td222gh}. It follows that $0^k\ggt_{n-1}01^{n-1}\in \StdExtRa(0^k\ggt_{n-1}01)\subseteq \rw$. Also we have that $0^k\ggt_{n-1}01\not= \StdExtR(0^k\ggt_{n-1}0)$ and thus $0^k\ggt_{n-1}01\in\bar \flexpref_n$.

Obviously $\occur(0^k\ggt_{n-1}01^n, 1^n)=1$. Since $1^n$ is a palindrome we have that $0^k\ggt_{n-1}01^n\in \rw$; see Proposition \ref{oij5498fr654td222gh}.
Since $\ggt_{n-1}01^n0\ggt_{n-1}$ is a palindrome we have that $\lps(0^k\ggt_{n-1}01^nt)=t^R1^nt$ for each $t\in \Prefix(0\ggt_{n-1})$. This implies that
$0^k\ggt_{n-1}01^nt\in \StdExtRa(0^k\ggt_{n-1}01^n)\subseteq \rw$ and in particular $0^k\ggt_{n}\in \StdExtRa(0^k\ggt_{n-1}01^n)\subseteq \rw$. 
Clearly $0^k\ggt_{n-1}01^{n}\not = \StdExtR(0^k\ggt_{n-1}01^{n-1})$ and thus $0^k\ggt_{n-1}01^{n}\in\bar \flexpref_n$.

Consequently for each $n,k\geq 2$, we conclude that $0^k\ggt_{n}\in \rw$ and $\bar \flexpref_n=\{0^k\ggt_{n-1}01, 0^k\ggt_{n-1}01^{n}\}\mbox{.}$
\end{proof}

We present all switches of $0^k\ggt_n$. Let $\swtop_n=\left(\switch(0^k\ggt_n)\setminus\switch(0^k\ggt_{n-1})\right)\cap\bigcup_{w\in \bar \flexpref_n}\Suffix(w)$, where $n\geq 3$.
\begin{proposition}
\label{gg5r84d66ee65r}  
If $k\geq 2$ and $n\geq 3$ then 
\[\swtop_n=\{00\ggt_{n-1}01, 01^{n-1}0\ggt_{n-2}01^n, 01^n\}\mbox{.}\]
\end{proposition}
\begin{proof}
Proposition \ref{hjf568sty} states that $\bar \flexpref_n=\{0^k\ggt_{n-1}01, 0^k\ggt_{n-1}01^{n}\}$. We will consider the switches that are suffixes of the flexed points from $\bar \flexpref_n$: 
\begin{itemize} 
\item
For $0^k\ggt_{n-1}01$: 
Let $t=\lps(0^k\ggt_{n-1}0)$. Obviously $t=0\ggt_{n-1}0$. Since $\occur(t,1^{n-1})=1$ it follows that $t$ is the only palindromic suffix of $0^k\ggt_{n-1}0$ which contains the factor $1^{n-1}$. Consequently each palindromic suffix of $0^k\ggt_{n-1}0$ which is not equal to $t$ is a factor of $0\ggt_{n-2}0\in \Suffix(t)$. Thus $00\ggt_{n-1}01$ is the only switch of $\bar t=00\ggt_{n-1}01$ which is not a switch of $0\ggt_{n-2}0\in\Factor(0^k\ggt_{n-1})$. 
\item
For $0^k\ggt_{n-1}01^{n}$:
Let $t\in \Suffix(0^k\ggt_{n-1}01^n)\cap \switch(0^k\ggt_{n})$.
Since $1^n\in \Suffix(00\ggt_{n-1}01^n)$ it follows that $\vert t\vert\geq n+1$. For $\vert t\vert=n+1$ there is the switch $01^n$. For $\vert t\vert>n+1$ we have that $1^{n-1}\in \Suffix(\rtrim(t))\cap\Prefix(\ltrim(t))$ and because $\occur(00\ggt_{n-1}01^{n-1}, 1^{n-1})=2$ it follows that there is only one switch with $\vert t\vert>n+1$; namely $\bar t=01^{n-1}0\ggt_{n-2}01^n$.  
\end{itemize}
The proposition follows.
\end{proof}
Proposition \ref{gg5r84d66ee65r} and Corollary \ref{mnr5r152riq4} allow us to list all switches of $0^k\ggt_n$.
\begin{corollary}
\label{iu1g21s2ww236}
If $n\geq 3$ then
\[\begin{split}
\switch(0^k\ggt_n)=\bigcup_{i=1}^k\{00^i1\}\cup \{01, 10, 00101, 11010, 01011\}\cup \\ \bigcup_{i=3}^n\{00\ggt_{n-1}01, 01^{n-1}0\ggt_{n-2}01^n, 1^n0\ggt_{n-2}01^{n-1}0, 01^n, 1^n0\} \mbox{.}
\end{split}
\]
\end{corollary}
\begin{proof}
Proposition \ref{gg5r84d66ee65r} states that $\swtop_n=\{00\ggt_{n-1}01, 01^{n-1}0\ggt_{n-2}01^n, 01^n\}$ for $n\geq 3$. We may easily check that \begin{itemize}\item $(00\ggt_{n-1}01)^R=10\ggt_{n-1}00\not\in\Factor(0^k\ggt_n)$, \item $(01^{n-1}0\ggt_{n-2}01^n)^R=1^n0\ggt_{n-2}01^{n-1}0\in \Factor(0^k\ggt_n)$, and \item $(00\ggt_{n-1}01)^R\not\in \Factor(0^k\ggt_m)$ for all $m\geq 2$.\end{itemize}

Obviously $\switch(0^k\ggt_2)=\bigcup_{i=1}^k\{00^i1\}\cup\{01, 10, 001, 00101, 01011\}$; recall that $0^k\ggt_2=0^k101101$. Note that $(01011)^R=11010\in\Factor(\ggt_3)$, $(00^i1)^R=10^i0\not \in \Factor(\ggt_m)$, and $(00101)^R=10100\not \in \Factor(\ggt_m)$ for all $i,m\geq 1$. Corollary \ref{mnr5r152riq4} asserts for every switch $t$ of $w$ that there is a flexed points $\bar w\in \flexpref(w)$ such that $\{t,t^R\}\cap \Suffix(\bar w)\not =\emptyset$. The corollary follows.
\end{proof}

Let $j\geq 2$. We define: \begin{itemize} \item $\alpha_{1,j}=00\ggt_{j-1}00$,
\item
$\alpha_{2,j}=01^{j-1}0\ggt_{j-2}01^{j-1}0$, \item
$\alpha_{3,j}=1^{j}0\ggt_{j-2}01^{j}$, and \item
$\alpha_{4,j}=1^{j+1}$. \end{itemize}

The next obvious corollary of Corollary \ref{iu1g21s2ww236} presents the switch palindromic closures of all switches of the word $0^k\ggt_n$.
\begin{corollary}
\label{trlkj555j12}
If $k\geq 2$ and $n\geq 3$ then 
\[
\begin{split}\swc(\switch(0^k\ggt_n))= \{0^k,00100,11011, 01010, \alpha_{4,n}\} \cup \\ \bigcup_{i=3}^n\{\alpha_{1,j}, \alpha_{2,j}, \alpha_{3,j}\}\mbox{.}
\end{split}
\]
\end{corollary}
\begin{proof}
Corollary \ref{iu1g21s2ww236} lists all switches of the word $0^k\ggt_n$. For every switch $t\in\switch(0^k\ggt_n)$ we may easily verify that there is $v\in \swc(\switch(0^k\ggt_n))$ such that $\swc(t)\in \Factor(v)$. This completes the proof.
\end{proof}
\begin{remark}
Note that the palindromes $\alpha_{4,j}$ are factors of $\alpha_{4,n}$ for $j\leq n$. This is the difference to palindromes $\alpha_{i,j}$, where $i\in \{1,2,3\}$. For this reason the palindrome $\alpha_{4,j}$ is not involved in the union formula from $i=3$ to $n$.
\end{remark}

The next definition defines a word $\ggh_n$. Later we show that $\ggh_n$ is a unique rich extension of $\bar \ggh_n$, where $\ggh_n=\bar \ggh_n\ltrim(\ggt_n)$.
\begin{definition}
Let $n\geq 3$. We define:
\begin{itemize}
\item
$\kappa(j,w)=\StdExtOLP_0(\StdExtWP(w, \alpha_{1,j}, \alpha_{2,j}, \alpha_{3,j}, \alpha_{4,j}))$, where $w\in \rw\cap \Alphabet^+$ and $3\leq j\leq n$.
\item
$\ggh_{n,n}=\kappa(n,000\ggt_n00\ggt_n)$.
\item
$\ggh_{n,j}=\kappa(j,\ggh_{n,j+1})$, where $3\leq j< n$.
\item
Suppose that $\Alphabet$ is totally ordered, let $\sigma(\Alphabet)=x_1x_2\dots x_m$, where $x_i\in \Alphabet\setminus\{0,1\}$, $x_i<x_{i+1}$, $1\leq i< m$, and $m=\AlphabetCard-2$.
\item
$\ggh_n=\sigma(\Alphabet)\StdExtWP(\ggh_{n,3}, 00100,11011,01010)$.
\end{itemize}
\begin{remark}
The function $\kappa(j,w)$ extends the word $w$ to a word $\bar ww$ in such a way that $\bar ww$ contains the switch palindromic closures of switches $\alpha_{1,j}$, $\alpha_{2,j}$, $\alpha_{3,j}$, $\alpha_{4,j}$. In addition the longest palindromic prefix of $\bar ww$ is  $0^k$ for some $k>0$.

The word $\ggh_{n,3}$ is constructed by iterative applying of the function $\kappa(j,w)$ starting with the word $000\ggt_n00\ggt_n$.

The word $\ggh_n$ contains the switch palindromic closures of all switches of the word $00\ggt_n$. The suffix of $\ggh_n$ is the word $000\ggt_n00\ggt_n$. As a result $\ggh_n$ has the form $u\rtrim(\ggt_n)1001\ltrim(\ggt_n)$ for some $u\in \Alphabet^*$. It is the form used in Proposition \ref{tg512c12vcc25e9}. 
The prefix $\sigma(\Alphabet)$ of $\ggh_n$ is there to assert that $u$ contains all letters. The order of the letters does not matter.
\end{remark}
We show that $\ggh_n$ is a rich word.
\end{definition}
\begin{lemma}
If $n\geq 3$ then $\ggh_n\in \rw$.
\end{lemma}
\begin{proof}
Lemma \ref{ttf56r25d2f} says that both $\StdExtWP(w,t), \StdExtOLP_0(w)\in \rw$, where $w,t\in \rw$. Proposition \ref{kkmnd5s8s658} and Proposition \ref{hjf568sty} imply that $\rtrim(\alpha_{i,j})\in \rw$, since $\rtrim(\alpha_{i,j})\in \Factor(00\ggt_j)\subseteq \Factor(00\ggt_n)$, where $i\in \{1,2,3,4\}$ and $3\leq j\leq n$.
Because $\alpha_{i,n}=\StdExtR(\rtrim(\alpha_{i,n}))$ we have that $\alpha_{i,n}\in \rw$, see Lemma \ref{fjis5e64s8e6jj5}. Hence $\kappa(j,w)\in \rw$. 

Proposition \ref{hjf568sty} asserts that $0^k\ggt_n\in \rw$. Also it is easy to see that $000\ggt_n00\ggt_n\in \rw$; just consider that $00\ggt_n00\ggt_n\in \StdExtRa(00\ggt_n)$, \[\occur(000\ggt_n00\ggt_n,000)=1\mbox{, and }\lpp(000\ggt_n00\ggt_n)=000\mbox{,}\] see Corollary \ref{odd58e62n2so}.
In consequence $\ggh_{n,j}\in \rw$ for $3\leq j<n$. 
We have that $\StdExtWP(\ggh_{n,3}, 00100,11011,01010)\in \rw$, because $00100,01010,11011\in \rw$. 

Obviously $\sigma(\Alphabet)\in \rw$. Moreover
it is easy to verify that if $w_1,w_2\in \rw$ and $\Factor(w_1)\cap\Factor(w_2)=\epsilon$ then $w_1w_2\in \rw$. Hence \[\sigma(\Alphabet)\StdExtWP(\ggh_{n,3}, 00100,11011,01010)\in \rw\mbox{.}\]
We conclude that $\ggh_n\in \rw$.
\end{proof}

\begin{proposition}
\label{tuj5d68e5f685sttu8}
Let $\bar \ggh_n$ be such that $\ggh_n=\bar \ggh_n \ltrim(\ggt_n)$. If $n>2$ then $\ggh_n$ is a unique rich extension of $\bar \ggh_n$.
\end{proposition}
\begin{proof}
Obviously there is $w\in \rw$ such that $\ggh_n=w\rtrim(\ggt_n)1001\ltrim(\ggt_n)$. 
Corollary \ref{trlkj555j12} lists the elements of $\swc(\switch(0^k\ggt_n))$. The construction of $\ggh_n$ guarantees that all these elements  are factors of $w$; formally $\alpha_{i,j}\in \Factor(w)$ for all $i\in \{1,2,3,4\}$ and $3\leq j\leq n$. For $00\ggt_20=001011010$ we can see that $\switch (001011010)=\{01,10,001,011,110,00101,01011,11010\}$. It follows that $\swc(\switch (001011010))=\{00100,01010,000,111,11011\}$. Obviously we have that $\{00100,01010,000,111,11011\}\subseteq \Factor(w)$.

Since $\sigma(\Alphabet)\in \Prefix(\ggh_n)$ we have that $\Alphabet\in \Factor(w)$. It is easy to check that $\lps(w\rtrim(\ggt_n)1001)=1001$. Hence we have \[w\rtrim(\ggt_n)1001\ltrim(\ggt_n)\in \StdExtRa(w\rtrim(\ggt_n)1001)\mbox{.}\]
Thus Proposition \ref{tg512c12vcc25e9} implies that $\ggh_n$ is a unique rich extension of $\bar \ggh_n$.
\end{proof}

Let $\rho(n)=\vert \ggt_n\vert$, where $n\geq 1$. Since $\ggt_n=\ggt_{n-1}01^n0\ggt_{n-1}$, we have $2\rho(n)<\rho(n+1)$ and consequently $\rho(n)< \frac{1}{2^{k}}\rho(n+k)$, where $k>0$. 

We derive an upper bound for length of $\ggh_n$. We start with an upper bound for $\vert\kappa(j,w)\vert$.

\begin{proposition}
\label{trtrtr44g65w6r}
If $j,k>2$, $\bar w\in \rw$, $w=0^k\ggt_{j}\bar w\in \rw$, $\lpp(w)=0^k$ and $\alpha_{i,j}\not \in \Factor(\bar w)$ for $i\in \{1,2,3\}$ then $\vert \kappa(j,w)\vert< \vert w\vert+7\rho(j-1)+5k+5j+10$.
\end{proposition}
\begin{proof}
Let $t_1=\StdExtWP(w,\alpha_{1,j})$, $t_2=\StdExtWP(t_1,\alpha_{2,j})$, $t_3=\StdExtWP(t_2,\alpha_{3,j})$, and $t_4=\StdExtWP(t_3,\alpha_{4,j})$. Clearly $\kappa(j,w)=\StdExtOLP_0(t_4)$.
It is easy to see that:
\begin{itemize}
\item $t_1=00\ggt_{j-1}0^k\ggt_j\bar w$; $\lpp(t_1)=\alpha_{1,j}=00\ggt_{j-1}00$.
\item $t_2=01^{j-1}0\ggt_{j-2}01^{j-1}0 \ggt_{j-2}0^{k-2}t_1$; \\ $\lpp(t_2)=\alpha_{2,j}=01^{j-1}0\ggt_{j-2}01^{j-1}0$.
\item $t_3=1^j0\ggt_{j-2}01^j0\ggt_{j-1}0^k\ggt_{j-1}0^k\ggt_{j-2}t_2$; $\lpp(t_3)=\alpha_{3,j}=1^j0\ggt_{j-2}01^j$.
\item If $\alpha_{4,j}\in \Factor(w)$ then $t_4=t_3$ and $\lpp(t_4)=\lpp(t_3)$ else $t_4=1t_3$ and $\lpp(t_4)=\alpha_{4,j}=1^{j+1}$.
\item
If $\alpha_{4,j}\in \Factor(w)$ then $\kappa(j,w)=00^k\ggt_{j-1}0t_4$ else \[\kappa(j,w)=00^k\ggt_{j-1}01^j0\ggt_{j-2}0t_4.\] In either case we have $\lpp(\kappa(j,w))=0^{k+1}$.
\end{itemize} 
\noindent
It follows that:
\begin{itemize}
\item
$\vert t_1\vert =\vert w\vert + \rho(j-1) +2$.
\item
$\vert t_2\vert = \vert t_1\vert + (k-2)+ 2\rho(j-2)+ 2(j-1)+4=\vert w\vert+6+\rho(j-1)+2\rho(j-2) + (k-2) + 2(j-1) $.
\item
$\vert t_3\vert=\vert t_2\vert + 2\rho(j-2) + 2\rho(j-1)+2k+3+2j=\vert w\vert+9+3\rho(j-1)+4\rho(j-2) + (k-2) +2k + 2(j-1)+2j$.
\item
$\vert t_4\vert \leq \vert t_3\vert+1$.
\item
$\vert \kappa(j,w)\vert \leq \vert t_4\vert + k+4 + \rho(j-1) + \rho(j-2)+ j=\vert w\vert+14+4\rho(j-1)+5\rho(j-2) + (k-2) +3k + 2(j-1)+3j$.
\end{itemize}
Since $2\rho(j-2)<\rho(j-1)$ we have $\vert \kappa(j,w)\vert<\vert w\vert+7\rho(j-1)+4k+5j+10$. 
\end{proof}
The main theorem of the section presents an upper bound for the length of $\ggh_n$.
\begin{theorem}
\label{tjs5d6r85rrr2f1}
If $n\geq 2$ then $\vert \ggh_n\vert< \frac{11}{2}\rho(n)+(n-3)(5n+22)+3n+20+\AlphabetCard$.
\end{theorem}
\begin{proof}
Proposition \ref{trtrtr44g65w6r} implies for $j=n$, $k=3$ and $w=000\ggt_n00\ggt_n$ that
\begin{equation}
\label{hjeq5585f5}
\vert\ggh_{n,n}\vert=\vert\kappa(n,w)\vert< \vert w\vert+7\rho(n-1)+4*3+5n+10\mbox{.} 
\end{equation}
For $n-1$ and $n-2$ we have:
\begin{itemize}
\item
$\vert\ggh_{n,n-1}\vert=\vert\kappa(n-1,\ggh_{n,n})\vert<\vert \ggh_{n,n}\vert+7\rho(n-2)+4*4+5(n-1)+10$.
\item
$\vert\ggh_{n,n-2}\vert=\vert\kappa(n-2),\ggh_{n,n-1}\vert<\vert \ggh_{n,n-1}\vert+7\rho(n-3)+4*5+5(n-2)+10$.
\end{itemize}
And generally for $n-i$:
\begin{equation}
\label{udjeqljios65498t}
\begin{split}
\vert\ggh_{n,n-i}\vert<\vert \ggh_{n,n-i+1}\vert+7\rho(n-i-1)+4(i+3)+5(n-i)+10= \\ \vert \ggh_{n,n-i+1}\vert+7\rho(n-i-1)+5n-i+22 < \\ \vert \ggh_{n,n-i+1}\vert+7\rho(n-i-1)+5n+22 \mbox{.}
\end{split}
\end{equation}
Realize that $\rho(n-i-1)\leq \frac{1}{2^{i+1}}\rho(n)$, $\vert w\vert=2\rho(n)+5$, and $\sum_{i=1}^{n-3}\frac{1}{2^{i+1}}< \frac{1}{2}$. It follows from (\ref{hjeq5585f5}) and (\ref{udjeqljios65498t}) that:
\begin{equation}
\label{hj566236w5e8}
\begin{split}
\vert \ggh_{n,3}\vert<2\rho(n)+5+7\rho(n)\sum_{i=1}^{n-3}\frac{1}{2^{i+1}}+(5n+22)\sum_{i=1}^{n-3}1< \\ \frac{11}{2}\rho(n)+(n-3)(5n+22)+5\mbox{.}
\end{split}
\end{equation}
Obviously $\lpp(\ggh_{n,3})=0^n$ and $0^n\ggt_20\in \Prefix(\ggh_{n,3})$. 

Let $t_1=\StdExtWP(\ggh_{n,3},00100)$, $t_2=\StdExtWP(t_1,11011)$. We have that $\ggh_n=\sigma(\Alphabet)\StdExtWP(t_2,01010)$. We can verify that 
$t_1=001\ggh_{n,3}$, $t_2=11011010^{n-2}t_1$, and $\ggh_n=\sigma(\Alphabet)01010^n10^n10t_2$.

Using (\ref{hj566236w5e8}) we get: 
 \[
 \begin{split}\vert \ggh_n\vert\leq\vert\ggh_{n,3}\vert+\vert 001\vert + \vert11011010^{n-2}\vert+\vert01010^n10^n10\vert+\AlphabetCard< \\ \frac{11}{2}\rho(n)+(n-3)(5n+22)+3n+20+\AlphabetCard\mbox{.}
 \end{split}\]
 This completes the proof.
\end{proof}
Theorem \ref{tjs5d6r85rrr2f1} and Proposition \ref{tuj5d68e5f685sttu8} have the following corollary to the lower bound for $\phi(n)$.
\begin{corollary}
For each real constant $c>0$ and each integer $m>0$ there is $n>m$ such that $\phi(n)\geq (\frac{2}{9}-c)n$. 
\end{corollary}
\begin{proof}
Proposition \ref{tuj5d68e5f685sttu8} implies that $\omega(\vert \bar \ggh_n\vert)\leq\vert \ggt_n\vert-1=\rho(n)-1$. It follows that $\phi(\vert \bar \ggh_n\vert)\geq \rho(n)-1$ and \begin{equation}
\label{ttks5w6re822s5e}
\phi(\vert \bar \ggh_n\vert)\geq \frac{\rho(n)-1}{\vert \bar \ggh_n\vert}\vert \bar \ggh_n\vert\mbox{.}\end{equation}

From Theorem \ref{tjs5d6r85rrr2f1} and Proposition \ref{tuj5d68e5f685sttu8} we have that \[
\frac{\rho(n)-1}{\vert\bar \ggh_n\vert}=\frac{\rho(n)-1}{\vert\ggh_n\vert-\rho(n)-1}=\frac{\rho(n)-1}{\frac{9}{2}\rho(n)+(n-3)(5n+22)+\AlphabetCard+4}.\]
Since $\rho(n)\geq 2^n$ this implies that \begin{equation}\label{tthjs523d15s4d}\frac{\rho(n)-1}{\vert\bar \ggh_n\vert}\leq\frac{2}{9}\mbox{ for }n>3\end{equation}  and \begin{equation}\label{jkj2a36e889t}\lim_{n\rightarrow\infty}\frac{\rho(n)-1}{\vert\bar \ggh_n\vert}=\frac{2}{9}\end{equation} The corollary follows from (\ref{ttks5w6re822s5e}),(\ref{tthjs523d15s4d}), and (\ref{jkj2a36e889t}).
\end{proof}
\section*{Acknowledgments}
The author acknowledges support by the Czech Science
Foundation grant GA\v CR 13-03538S and by the Grant Agency of the Czech Technical University in Prague, grant No. SGS14/205/OHK4/3T/14.

\bibliographystyle{siam}
\IfFileExists{biblio.bib}{\bibliography{biblio}}{\bibliography{../!bibliography/biblio}}

\begin{thebibliography}{1}

\bibitem{BuLuGlZa2}
{\sc M.~Bucci, A.~{De Luca}, A.~Glen, and L.~Q. Zamboni}, {\em A new
  characteristic property of rich words}, Theor. Comput. Sci., 410 (2009),
  pp.~2860--2863.

\bibitem{DrJuPi}
{\sc X.~Droubay, J.~Justin, and G.~Pirillo}, {\em Episturmian words and some
  constructions of de {L}uca and {R}auzy}, Theor. Comput. Sci., 255 (2001),
  pp.~539--553.

\bibitem{GlJuWiZa}
{\sc A.~Glen, J.~Justin, S.~Widmer, and L.~Q. Zamboni}, {\em Palindromic
  richness}, Eur. J. Combin., 30 (2009), pp.~510--531.

\bibitem{RukavickaRichWords2019}
{\sc J.~Rukavicka}, {\em Rich words containing two given factors}, Mercaş R.,
  Reidenbach D. (eds) Combinatorics on Words. WORDS 2019. Lecture Notes in
  Computer Science, vol 11682. Springer, Cham, DOI:
  https://doi.org/10.1007/978-3-030-28796-2\_23.

\bibitem{Vesti2014}
{\sc J.~Vesti}, {\em Extensions of rich words}, Theor. Comput. Sci., 548
  (2014), pp.~14--24.

\end{thebibliography}

\end{document}